\documentclass[11pt]{article}

\usepackage[letterpaper,margin=1in]{geometry}
\usepackage[T1]{fontenc}
\usepackage{lmodern}
\usepackage{amsmath,amssymb,amsthm,mathtools}
\usepackage{booktabs,array}
\usepackage{graphicx}
\usepackage{microtype}
\usepackage{enumitem}
\usepackage{xcolor}
\usepackage{hyperref}
\usepackage{caption}

\hypersetup{
  colorlinks=true,
  linkcolor=black,
  citecolor=black,
  urlcolor=blue!55!black,
  pdftitle={Independence-System Realisations in Single-Source Unsplittable Flow},
  pdfauthor={Koyar Afrasyab}
}

\newtheorem{theorem}{Theorem}
\newtheorem{corollary}[theorem]{Corollary}
\newtheorem{lemma}[theorem]{Lemma}

\theoremstyle{definition}
\newtheorem{definition}[theorem]{Definition}

\newcommand{\R}{\mathbb{R}}
\newcommand{\calI}{\mathcal{I}}
\newcommand{\calF}{\mathcal{F}}

\setlist[itemize]{leftmargin=1.6em,itemsep=2pt,topsep=4pt}
\setlist[enumerate]{leftmargin=1.8em,itemsep=2pt,topsep=4pt}
\title{\textbf{Independence-System Realisations in Single-Source Unsplittable Flow}}
\author{Koyar Afrasyab\\
\small\href{mailto:koyar@kinvectum.com}{koyar@kinvectum.com}}
\date{27 July 2026}

\begin{document}
\maketitle

\begin{abstract}
Additive-congestion constraints in single-source unsplittable flow can enforce stable-set structure. This note isolates and generalises that mechanism. We introduce a path-closed notion of realising an independence system by the zero-cost choices of primary terminals in a directed acyclic flow instance. The definition quantifies over every directed source--terminal path and therefore remains valid under prefix borrowing, suffix splicing, and hybrid routes.

Our main result extends the triangle mechanism: every finite loopless independence system has a polynomial-size realisation, measured in the incidence size of its minimal forbidden sets. Hence every finite simple graph, and more generally every hypergraph independence system without singleton forbidden hyperedges, is representable by an acyclic single-source gadget. We then specialise the construction to odd cycles. For \(C_{2k+1}\), a uniform rational family produces a fractional cheap-selection vector that violates the odd-cycle inequality. A potential shift converts a signed connector separator into nonnegative arc costs and gives the exact cost-preserving additive-congestion threshold
\[
\tau=1-bq.
\]
Within the symmetric family, the supremum threshold is
\(\frac{k+2}{2(k+1)}\), which tends to \(1/2\).
For \(C_5\), an exact certificate independently derives all source--terminal paths and enumerates all \(3^{10}=59{,}049\) unsplittable routings using rational arithmetic.
\end{abstract}

\section{Introduction}

In the single-source unsplittable-flow problem, a common source supplies several terminals with unrelated positive demands. A fractional flow may split a terminal's demand over several paths, whereas an unsplittable routing must choose a single source--terminal path for each demand. The classical theorem of Dinitz, Garg, and Goemans shows that a feasible fractional flow can always be rounded while increasing the load on each arc by at most the maximum demand \cite{dgg1999}. Goemans subsequently conjectured that the same additive guarantee could be achieved without increasing total cost.

On 22 July 2026, Dmitry Rybin publicly announced a seven-vertex counterexample obtained in a shared GPT-5.6 Pro investigation \cite{rybin2026,rybintranscript2026}. Its three demands are \(15,10,15\); the exhibited fractional flow has cost \(58\), while every additive-\(15\)-good unsplittable routing has cost at least \(60\). The machine-readable copy in this repository is a reproduction of that instance, not an original construction of this article. Our independent exact verifier derives all six source--terminal paths from the arc list and checks all eight routings. Thus the refutation used here follows from the reproduced finite certificate and its independent exact audit, not from the authority or eventual status of the announcement.

The certificate exhibits a triangle conflict mechanism and motivates a structural question: which discrete choice systems can be encoded by the additive-congestion inequalities of a single-source instance? A primary terminal is interpreted as ``selected'' when it uses its unique zero-cost path. Auxiliary terminals enforce incompatibilities through ordinary shared arc loads. The central difficulty is path closure: a valid construction cannot prescribe a short menu of routes and ignore other paths that the graph itself creates. Any source prefix may combine with any reachable suffix, so a sound proof must survive all borrowed and hybrid routes.

\subsection{Contributions}

The main contributions are as follows.

\begin{enumerate}
  \item We formalise the mechanism through \emph{strong realisation}, which quantifies over every actual directed path in the gadget rather than over a prescribed path decomposition.
  \item We give an explicit acyclic construction that strongly realises every finite loopless independence system. Its size is linear in the ground-set size plus the total incidence size of the minimal forbidden sets.
  \item We derive a uniform odd-cycle family whose fractional cheap-selection vector violates the stable-set odd-cycle inequality.
  \item We construct nonnegative arc costs and determine the exact cost-preserving additive-congestion threshold of the family.
  \item We independently certify the provenance-labelled triangle instance and supply an exact \(C_5\) certificate, using all-path, all-routing verifiers based solely on rational arithmetic.
\end{enumerate}

\subsection{Scope}

Rybin's triangle already refutes the additive-\(D\) cost conjecture. The odd-cycle family here is not a stronger refutation: its exact threshold is strictly below \(D\) and approaches \(D/2\). Its purpose is to generalise the underlying stable-set mechanism and to quantify one symmetric family. The optimisation statement in Section~\ref{sec:optimisation} is scoped to that family; it neither reaches \(D\) nor rules out stronger unrelated gadgets. The universal construction is elementary and closely related to standard conflict-encoding ideas. We therefore present this work as a research note and do not claim broad priority for the general mechanism absent a comprehensive expert literature review.

\section{Related work}

Dinitz, Garg, and Goemans proved that every single-source fractional flow admits an unsplittable routing whose arc load is at most the fractional load plus the largest demand \cite{dgg1999}. Morell and Skutella developed a short proof and a lower-bound counterpart, and formulated a stronger simultaneous-deviation perspective \cite{morell2022}. Majthoub Almoghrabi, Skutella, and Warode proved strong decomposition results for series-parallel digraphs \cite{almoghrabi2025}. Swamy, Traub, Vargas Koch, and Zenklusen showed how an unweighted error-bounded variant implies a relaxed cost guarantee \cite{swamy2026}.

Traub, Vargas Koch, and Zenklusen proved that planar acyclic instances admit the simultaneous upper and lower additive-\(D\) bounds without costs, and the cost-preserving conclusion with additive-\(2D\) bounds \cite{traub2026}. The reproduced seven-vertex graph is itself planar: its underlying undirected graph suppresses to \(K_4\) after the three degree-two terminals are removed. The exact checker verifies this subdivision certificate. Consequently, the audited instance rules out the cost-preserving additive-\(D\) statement even for planar acyclic graphs, but it does not contradict either the cost-free additive-\(D\) theorem or the cost-preserving additive-\(2D\) theorem.

The publicly announced triangle instance is the immediate antecedent of this note \cite{rybin2026,rybintranscript2026}. Its shared spine enforces a triangle conflict system: each zero-cost detour represents a selected vertex, and any pair overloads one spine arc under the additive-\(D\) bounds. Our construction replaces the nested triangle spine with one incidence resource per element--forbidden-set pair and one auxiliary demand per minimal forbidden set. This attribution identifies the methodological seed; the results below concern the resulting general construction and its certificates.

The gadgets also touch graph-representation and conflict-flow literatures. EPT and EPG representations define adjacency through intersections among chosen paths \cite{golumbic1985}. Minimum-cost flow with conflicts adds an external conflict relation on arcs \cite{suvak2021}. Here the forbidden system is induced internally by fractional baseline loads, auxiliary demands, and additive arc inequalities. Encoding conflicts with shared capacities is a familiar modelling pattern, so the contribution claimed here is limited to the stated path-closed construction, proof, threshold calculation, and exact certificates. A targeted search found no directly equivalent theorem, but that search is not an exhaustive priority determination.

\section{Model and definitions}\label{sec:model}

Let \(G=(V,A)\) be a finite directed acyclic graph with source \(s\), terminal set \(T\), and positive terminal demands \(d_t\). A feasible fractional flow \(x\in\R_{\ge0}^{A}\) has divergence \(-\sum_{t\in T}d_t\) at \(s\), divergence \(d_t\) at terminal \(t\), and zero divergence at every other vertex. Arc costs are nonnegative and measured per unit of flow: \(c\in\R_{\ge0}^{A}\). Let
\[
D:=\max_{t\in T}d_t.
\]

An unsplittable routing \(P=(P_t)_{t\in T}\) chooses a directed \(s\)--\(t\) path for each terminal. Its load on an arc \(a\) is
\[
y_P(a)=\sum_{t:\,a\in P_t}d_t,
\]
and its cost is \(c^Ty_P\).

\begin{definition}[Additive goodness]
For \(\lambda\ge0\), a routing \(P\) is \(\lambda\)-good if
\[
y_P(a)\le x(a)+\lambda D\qquad\text{for every }a\in A.
\]
\end{definition}

We distinguish \emph{primary} terminals from auxiliary terminals. A primary terminal is \emph{cheaply selected} when its chosen path has total cost zero.

\begin{definition}[Strong realisation]
Let \(\calI\subseteq2^E\) be an independence system on a finite ground set \(E\). A flow instance strongly realises \(\calI\) when:
\begin{enumerate}
  \item primary terminals are indexed by \(E\);
  \item every primary terminal has exactly one zero-cost path; and
  \item for every \(S\subseteq E\), there exists a \(1\)-good routing that cheaply selects exactly \(S\) if and only if \(S\in\calI\).
\end{enumerate}
The quantifier ranges over all directed paths in \(G\).
\end{definition}

An independence system is \emph{loopless} when every singleton belongs to it. Equivalently, every inclusion-minimal forbidden set has cardinality at least two. Write \(\calF\) for the antichain of minimal forbidden sets. Then
\[
S\in\calI
\quad\Longleftrightarrow\quad
F\nsubseteq S\ \text{ for every }F\in\calF.
\]

\section{A universal realisation theorem}\label{sec:universal}

We now give the construction in explicit vertex-and-arc form. This makes the acyclicity, path closure, and fractional feasibility transparent.

\subsection{Construction}

Fix \(i\in E\), and order the forbidden sets containing \(i\) as
\[
F_{i,1},F_{i,2},\ldots,F_{i,m_i}.
\]
When \(m_i\ge1\), introduce private vertices \(u_{i,j},v_{i,j}\) for \(j=1,\ldots,m_i\). Add:
\begin{itemize}
\item a start arc \(s\to u_{i,1}\);
\item incidence resources \(r_{F_{i,j},i}:u_{i,j}\to v_{i,j}\);
\item connectors \(v_{i,j}\to u_{i,j+1}\) for \(j<m_i\);
\item a final arc \(v_{i,m_i}\to t_i\);
\item a direct arc \(s\to t_i\); and
\item approaches \(s\to u_{i,j}\) for \(j\ge2\).
\end{itemize}
If \(m_i=0\), retain a private two-arc chain from \(s\) to \(t_i\) together with the direct arc.

For each \(F\in\calF\), create an auxiliary terminal \(a_F\). Whenever \(F=F_{i,j}\), add an exit
\[
v_{i,j}\to a_F.
\]
Thus any path ending at \(a_F\) has a final exit associated with some incidence \((F,i)\), and immediately before that exit it must traverse \(r_{F,i}\).

All chain, incidence, connector, final, and exit arcs have zero cost. Direct arcs and approaches have any strictly positive cost. The graph is acyclic: place \(s\) first, order the private chain vertices from left to right, and place all primary and auxiliary terminals last.

\subsection{Fractional flow}

Set \(q=1/4\), and give every terminal demand \(1\), so \(D=1\). For each primary \(i\), send \(q\) units along its full chain and \(1-q\) along its direct arc. For each auxiliary \(a_F\), split one unit uniformly among the \(|F|\) incidences. At a first incidence use the chain start; at a later incidence use its private approach. In either case, traverse only \(r_{F,i}\) and then exit to \(a_F\).

The important fractional loads are
\[
x(r_{F,i})=q+\frac1{|F|}<1
\]
and
\[
x(\text{connector})=q.
\]
All other loads follow directly from the described path flow, so flow conservation holds at every vertex.

\begin{lemma}[Exit dominator]\label{lem:dominator}
Every directed \(s\)--\(a_F\) path contains an incidence resource \(r_{F,i}\) immediately before its final exit.
\end{lemma}

\begin{proof}
The only arcs entering \(a_F\) are the exits \(v_{i,j}\to a_F\) for incidences with \(F_{i,j}=F\). The only arcs entering the corresponding vertex \(v_{i,j}\) are the resource arcs \(r_{F,i}:u_{i,j}\to v_{i,j}\). Hence any path using that exit must first traverse \(r_{F,i}\). A borrowed prefix may add earlier resources but cannot bypass this last one.
\end{proof}

\begin{lemma}[Unique cheap primary path]\label{lem:unique}
For each primary terminal \(t_i\), its full private chain is its unique zero-cost path.
\end{lemma}

\begin{proof}
The direct arc has positive cost. Any path entering the chain downstream uses a positive-cost approach. There are no arcs from another chain or from an auxiliary terminal into the chain, so the only remaining \(s\)--\(t_i\) path is the full chain, whose arcs all have zero cost.
\end{proof}

\begin{theorem}[Universal realisation]\label{thm:universal}
Every finite loopless independence system has a strong realisation in an acyclic single-source unsplittable-flow instance. The construction has size
\[
O\!\left(|E|+\sum_{F\in\calF}|F|\right).
\]
\end{theorem}

\begin{proof}
Lemmas~\ref{lem:dominator} and \ref{lem:unique} establish path closure and unique cheap selection.

For the forward direction, suppose a \(1\)-good routing selects every element of some \(F\in\calF\). The auxiliary \(a_F\) must use a final exit associated with some \(i\in F\). By Lemma~\ref{lem:dominator}, its path traverses \(r_{F,i}\). The selected primary \(i\) also traverses this resource, so its integral load is \(2\). However,
\[
x(r_{F,i})+D
=q+\frac1{|F|}+1
\le\frac14+\frac12+1
<2,
\]
contradicting \(1\)-goodness. Thus the selected set contains no minimal forbidden set and belongs to \(\calI\).

Conversely, take \(S\in\calI\). For every \(F\in\calF\), choose \(i_F\in F\) with \(i_F\notin S\). Route each selected primary along its full chain and each unselected primary directly. Route \(a_F\) locally through \(r_{F,i_F}\): use the start path when this is the first incidence of chain \(i_F\), and otherwise use the private approach. No auxiliary shares an incidence resource with a selected primary. Every resource load is therefore at most \(1\), every selected-chain arc has load at most \(1\), and every direct or approach arc has load at most \(1\). Since all fractional loads are nonnegative and \(D=1\), every arc load is at most \(x+D\). The routing is \(1\)-good and selects exactly \(S\).

There are two private vertices and a constant number of arcs per incidence, plus one terminal and a constant number of arcs per ground element or forbidden set, giving the stated size bound.
\end{proof}

\begin{corollary}\label{cor:graphs}
Every finite simple graph is strongly realisable by taking its edges as the minimal forbidden sets. More generally, every finite hypergraph independence system with no singleton forbidden hyperedge is strongly realisable.
\end{corollary}

This includes complete graphs, odd holes, odd antiholes, webs, antiwebs, and minimally imperfect graphs. The point is not that these graph classes share a path-intersection representation, but that their stable sets can be encoded by ordinary load inequalities in a single-source flow instance.

\section{A uniform odd-cycle family}\label{sec:cycle}

Fix \(k\ge1\), put \(n=2k+1\), and index vertices modulo \(n\). Let \(e_i=\{i,i+1\}\). On primary chain \(i\), order the incidence for \(e_{i-1}\) first, insert a connector \(h_i\), and place the incidence for \(e_i\) second. Add a private approach from \(s\) to the second incidence.

Choose rational
\[
0<\varepsilon<\frac1{4(k+1)}
\]
and define
\[
T=\frac{k}{2(k+1)},\qquad
t=T+\varepsilon,\qquad
s_0=\frac12+\varepsilon,
\]
\[
b=s_0+t,\qquad q=\frac{t}{b}.
\]
Primary demands are \(b\), edge-auxiliary demands are \(1\), and hence \(D=1\). Send \(t=bq\) units of each primary along its full chain and \(s_0=b(1-q)\) along its direct arc. Split each edge-auxiliary flow equally between the first incidence at \(i+1\) and the approached second incidence at \(i\). Every incidence resource has fractional load \(t+1/2\), and every connector has fractional load \(t\).

\begin{theorem}[Exact stable-set correspondence]\label{thm:cycle}
The \(1\)-good cheap selections of the odd-cycle instance are exactly the stable sets of \(C_n\).
\end{theorem}

\begin{proof}
If adjacent primaries \(i\) and \(i+1\) are both selected, the edge auxiliary for \(e_i\) must traverse an incidence resource of one endpoint. That resource has load \(b+1\), whereas its \(1\)-good allowance is
\[
t+\frac12+1.
\]
The excess is
\[
b+1-\left(t+\frac12+1\right)
=s_0-\frac12
=\varepsilon>0.
\]
Thus a selected set is stable.

Conversely, let \(S\) be stable. For each edge, route its auxiliary through an unselected endpoint, using the private approach when that endpoint's incidence is second. No selected primary shares an incidence resource with an auxiliary. Direct inspection of starts, connectors, approaches, exits, and direct arcs shows that every load is at most its fractional load plus \(D\). Hence the routing is \(1\)-good and selects exactly \(S\).
\end{proof}

The fractional cheap-selection coordinate of every primary is \(q\). Since
\begin{equation}\label{eq:violation}
nq-k=\frac{\varepsilon}{b}>0,
\end{equation}
the vector \((q,\ldots,q)\) violates the odd-cycle inequality
\[
\sum_{i=0}^{n-1}z_i\le k.
\]

\section{Nonnegative costs and an exact threshold}\label{sec:threshold}

To state the result precisely, define the cost-preserving additive-congestion threshold of an instance by
\[
\tau^\star
:=\inf\left\{\lambda\ge0:
\begin{array}{l}
\text{there exists a \(\lambda\)-good unsplittable routing \(P\)}\\[-2pt]
\text{with \(c^Ty_P\le c^Tx\)}
\end{array}\right\}.
\]

Start with a signed weight \(p\) that assigns \(-1/b\) to every connector \(h_i\) and \(0\) to every other arc. Give potential \(0\) to vertices before a connector and \(-1/b\) to vertices after it, including primary and auxiliary terminals. Define reduced costs
\[
c(u,v)=p(u,v)+\pi(u)-\pi(v).
\]
These costs are nonnegative. Equivalently, cost \(1/b\) is placed on each direct primary arc, each approach to a second incidence, and each exit from a first incidence; every other arc has cost zero.

Fractional and integral routings have identical divergence. Consequently the potential terms cancel:
\begin{equation}\label{eq:potential}
c^T(y-x)=p^T(y-x).
\end{equation}

\begin{theorem}[Exact cost-congestion threshold]\label{thm:threshold}
For the uniform odd-cycle family,
\[
\tau^\star=1-t=1-bq.
\]
More precisely, for every \(\lambda<1-t\) and every \(\lambda\)-good unsplittable routing \(y\),
\[
c^Ty\ge c^Tx+nq-k
=c^Tx+\frac{\varepsilon}{b}.
\]
At \(\lambda=1-t\), there is an unsplittable routing of cost at most \(c^Tx\).
\end{theorem}

\begin{proof}
A borrowed auxiliary route to a second incidence uses connector \(h_i\) and places demand \(1\) on it. Because \(x(h_i)=t\), such a route is infeasible whenever
\[
1>t+\lambda,
\]
that is, whenever \(\lambda<1-t\). In this range, connector \(h_i\) has integral load \(b\) exactly when primary \(i\) is selected, and zero otherwise.

If two adjacent primaries were selected, their edge auxiliary could not borrow a connector and would have to share an incidence resource with one endpoint. Its load would be \(b+1\), while
\[
t+\frac12+\lambda
<t+\frac12+1-t
=\frac32
<b+1.
\]
Hence the selected set is stable and has cardinality at most \(k\). Therefore
\[
p^Ty\ge-k,
\qquad
p^Tx=-nq.
\]
Equation~\eqref{eq:potential} and \eqref{eq:violation} give
\[
c^T(y-x)=p^T(y-x)\ge nq-k=\frac{\varepsilon}{b}.
\]

At \(\lambda=1-t\), route every primary directly. Route each edge auxiliary along the full borrowed prefix of one primary chain to its second incidence. Every connector then carries exactly
\[
1=t+(1-t),
\]
and all other bounds hold because \(t<1/2\). The routing costs \(n\), while
\[
c^Tx=n\left(1-q+\frac1b\right)>n.
\]
Thus a cost-preserving routing first appears at \(\lambda=1-t\).
\end{proof}

\section{Optimisation within the symmetric family}\label{sec:optimisation}

Write \(t=bq\) and \(s_0=b(1-q)\). Exact stable-set realisation requires \(s_0>1/2\), while violation of the odd-cycle inequality requires \(nq>k\). Since \(n=2k+1\),
\[
(k+1)t>ks_0>\frac{k}{2}.
\]
It follows that
\[
\tau^\star=1-t<\frac{k+2}{2(k+1)}.
\]
Our \(\varepsilon\)-family attains
\[
\tau^\star
=\frac{k+2}{2(k+1)}-\varepsilon.
\]
Therefore the exact supremum within this uniform symmetric parameterisation is
\[
\tau_k^{\mathrm{sup}}
=\frac{k+2}{2(k+1)}
=\frac12+\frac1{2(k+1)},
\qquad
\lim_{k\to\infty}\tau_k^{\mathrm{sup}}=\frac12.
\]

\section{An exact \(C_5\) certificate}\label{sec:c5}

For \(k=2\), choose \(\varepsilon=1/120\). Table~\ref{tab:c5} records the resulting parameters.

\begin{table}[ht]
\centering
\caption{Exact parameters of the \(C_5\) certificate.}
\label{tab:c5}
\begin{tabular}{@{}cccccc@{}}
\toprule
\(b\) & \(q\) & \(t=bq\) & \(\tau^\star\) & \(5q-2\) & \(c^Tx\)\\
\midrule
\(\frac{17}{20}\) &
\(\frac{41}{102}\) &
\(\frac{41}{120}\) &
\(\frac{79}{120}\) &
\(\frac1{102}\) &
\(\frac{905}{102}\)\\
\bottomrule
\end{tabular}
\end{table}

The machine-readable DAG has \(31\) vertices and \(45\) arcs. Every one of the five primary and five auxiliary terminals has exactly three actual paths, so the verifier enumerates
\[
3^{10}=59{,}049
\]
unsplittable routings. It derives the path sets from the arc list and does not trust a prescribed decomposition.

\begin{figure}[p]
\centering
\includegraphics[width=0.72\textwidth]{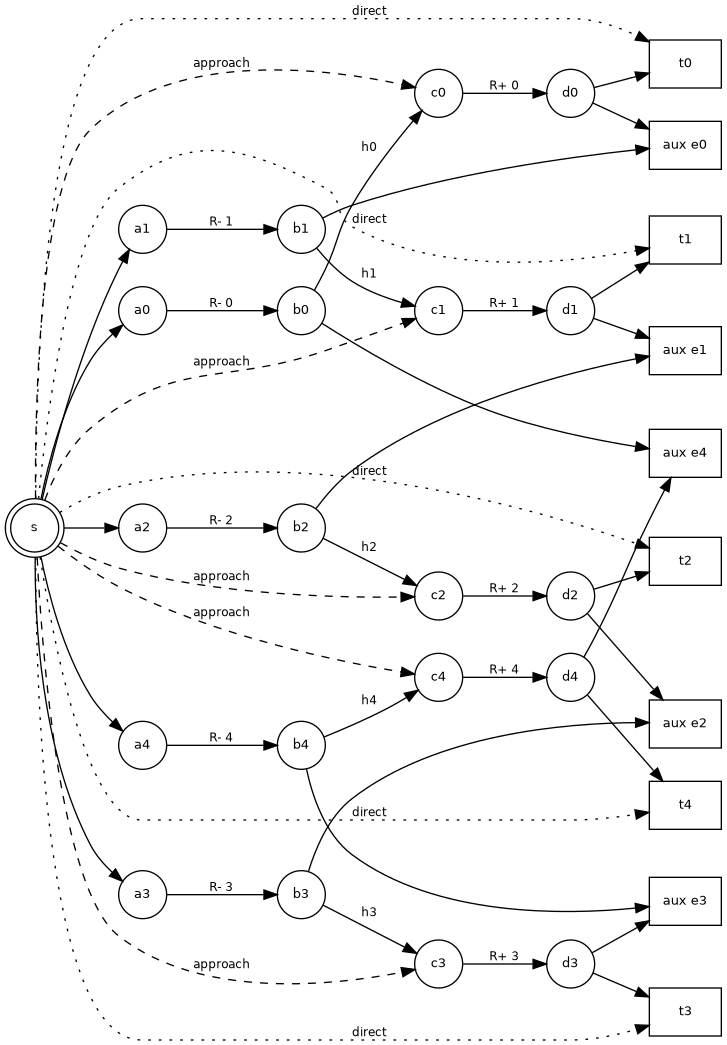}
\caption{The exact \(C_5\) DAG. Solid arcs are zero-cost chain, resource, connector, or exit arcs; dashed approaches and dotted direct arcs carry positive cost. Labels \(R^-_i\), \(h_i\), and \(R^+_i\) denote the first resource, connector, and second resource of primary chain \(i\).}
\label{fig:c5}
\end{figure}

The independent verifier establishes:
\begin{itemize}
  \item fractional feasibility and acyclicity;
  \item exactly one zero-cost path for each primary terminal;
  \item exactly the \(11\) stable sets of \(C_5\) as cheap masks at \(\lambda=1\);
  \item at \(\lambda=7/12\), minimum integral cost \(151/17\), exceeding \(c^Tx=905/102\) by exactly \(1/102\);
  \item at \(\tau^\star=79/120\), a routing of cost \(5<c^Tx\);
  \item rejection of mutations that remove an incompatibility, add a bypass, create a cheap hybrid primary route, or erase the odd-cycle violation.
\end{itemize}

All arithmetic uses Python's exact \texttt{Fraction} type. The finite certificate supports the general proof but does not replace it: exhaustive enumeration is restricted to \(C_5\), while the universal and odd-cycle theorems are symbolic.

\section{Discussion and limitations}

The universal construction shows that additive load inequalities can encode arbitrary finite independence systems, provided singleton selection remains feasible. This expressive power comes from three ingredients: a private zero-cost chain for each decision, an incidence resource for each participation in a minimal obstruction, and an auxiliary unit demand for each obstruction. The exit-dominator property makes the mechanism robust to path splicing.

The construction size is polynomial in an explicit minimal-forbidden-set description, not necessarily in a more compressed oracle or graph representation. An independence system can have exponentially many minimal forbidden sets, and the theorem does not avoid that representation cost.

The exact threshold is proved only for the uniform symmetric odd-cycle family. It does not imply that \(1/2\) is a universal barrier and, unlike Rybin's triangle, it does not reach the additive-\(D\) boundary. The family therefore supplies a structural generalisation and a scoped calculation rather than a stronger counterexample. The construction is short and the underlying conflict-encoding idea may be folklore; no claim of definitive novelty or flagship scope is made. Finally, the general proof has not been formalised in a proof assistant. The exact scripts machine-check the finite certificates and rational identities, while the universal theorem remains a conventional mathematical proof.

\section{Conclusion}

Additive-congestion inequalities can enforce a stable-set obstruction: every finite loopless independence system can be embedded into the cheap-choice structure of an acyclic single-source unsplittable-flow instance in a manner closed under all actual graph paths. Odd cycles turn the qualitative construction into a scoped cost-congestion calculation, and the \(C_5\) certificate demonstrates that the generalised phenomenon is reproducible at finite scale.

The construction suggests two natural directions. First, one may seek compressed gadgets for structured forbidden systems. Second, one may investigate which additional graph restrictions---planarity, bounded genus, bounded treewidth, or series-parallel structure---limit the independence systems that can be realised and thereby constrain cost-congestion lower bounds.

\appendix
\section{Reproducibility protocol}

The accompanying package contains a provenance-labelled reproduction of Rybin's triangle instance, the \(C_5\) JSON instance, the exhaustive verifiers, and the symbolic parameter checker. From the package directory, run
\begin{verbatim}
./verify_all.sh
\end{verbatim}
The suite enumerates all paths and routings of both finite certificates, checks the parameter identities and inequalities for \(k=1,\ldots,200\), and runs adversarial mutation tests. It requires no third-party Python package.

\IfFileExists{article.bbl}{

}{
  \bibliographystyle{plain}
  \bibliography{references}

\begin{thebibliography}{1}

\bibitem{dgg1999}
Yefim Dinitz, Naveen Garg, and Michel~X. Goemans.
\newblock On the single-source unsplittable flow problem.
\newblock {\em Combinatorica}, 19(1):17--41, 1999.

\bibitem{golumbic1985}
Martin~Charles Golumbic and Robert~E. Jamison.
\newblock The edge intersection graphs of paths in a tree.
\newblock {\em Journal of Combinatorial Theory, Series B}, 38:8--22, 1985.

\bibitem{almoghrabi2025}
Mohammed Majthoub~Almoghrabi, Martin Skutella, and Philipp Warode.
\newblock Integer and unsplittable multiflows in series-parallel digraphs.
\newblock In {\em Integer Programming and Combinatorial Optimization}, volume
  15620 of {\em Lecture Notes in Computer Science}, pages 427--441. Springer,
  2025.

\bibitem{morell2022}
Sarah Morell and Martin Skutella.
\newblock Single source unsplittable flows with arc-wise lower and upper
  bounds.
\newblock {\em Mathematical Programming}, 192(1):477--496, 2022.

\bibitem{rybintranscript2026}
Dmitry Rybin.
\newblock Counterexample to dinitz conjecture.
\newblock \url{https://chatgpt.com/share/6a60b2eb-0b64-83ee-9c76-7931ca1de063},
  July 2026.
\newblock Shared GPT-5.6 Pro investigation transcript; accessed 27 July 2026.

\bibitem{rybin2026}
Dmitry Rybin.
\newblock Dinitz--garg--goemans conjecture is false.
\newblock \url{https://x.com/DmitryRybin1/status/2079904005652893709}, July
  2026.
\newblock Public announcement on X; URL resolved through X's oEmbed service on
  27 July 2026.

\bibitem{suvak2021}
Zeynep {\c{S}}uvak, {\.I}~Kuban Alt{\i}nel, and Necati Aras.
\newblock Minimum cost flow problem with conflicts.
\newblock {\em Networks}, 78:421--442, 2021.

\bibitem{swamy2026}
Chaitanya Swamy, Vera Traub, Laura Vargas~Koch, and Rico Zenklusen.
\newblock Unsplittable cost flows from unweighted error-bounded variants.
\newblock In {\em Proceedings of the 2026 SIAM Symposium on Simplicity in
  Algorithms}, pages 512--523, 2026.

\bibitem{traub2026}
Vera Traub, Laura Vargas~Koch, and Rico Zenklusen.
\newblock Single-source unsplittable flows in planar and bounded-genus graphs.
\newblock {\em Mathematical Programming}, 2026.

\end{thebibliography}
}

\end{document}